\documentclass[a4paper]{article}
\usepackage[utf8]{inputenc}

\usepackage{RR}

\usepackage{hyperref}

\usepackage{listings}
\usepackage{graphicx}
\usepackage{float} 
\usepackage{amsmath} 
\usepackage{amsthm} 
\usepackage{color}

\newtheorem{definition}{Definition}
\newtheorem{theorem}{Theorem}
\newtheorem{example}{Example}
\newcommand{\Sa}{\Sigma}
\newcommand{\Sk}{\Sigma^{*}}
\newcommand{\e}{\varepsilon}

\RRdate{Decembre 2011}

\RRetitle{Abstract unordered and ordered trees CRDT\footnote{This work is a delivrable of french national research programs ConcRDanT (ANR-10-BLAN-0208).}}

\RRtitle{Arbres ordonn\'e et non ordonn\'e en CRDT\footnote{Ce travail est aussi un délivrable de l'ANR ConcRDanT (ANR-10-BLAN-0208).}}

\RRauthor{St\'ephane Martin  \and Mehdi Ahmed-Nacer \and Pascal Urso}

\authorhead{St\'ephane Martin \& Mehdi Ahmed-Nacer \& Pascal Urso}

\RRmotcle{Consistance \`a terme, CRDT, R\'eplication optimiste, Arbres,
  Consistance des donn\'ees}

\RRkeyword{Distributed System, Eventual Consistency, CRDT, Optimistic
  Replication, Data Consistency, Tree}

\RRresume{Les arbres sont une structure de donnée fondamentale dans
  beaucoup de domaines de l'informatique théorique et de l'ingénierie
  logicielle.  Dans ce rapport, nous montrons comment assurer la
  consistance d'arbres répliqués de manière optimiste.  Dans la
  réplication optimiste, les différentes répliques d'un système
  distribué peuvent passer par différents états intermédiaires avant
  de converger. Une nouvelle méthode pour assurer la convergence est
  de définir des CRDT (Conflict-free Replicated Data Types).

Dans ce rapport, nous proposons une collection de CRDT structure
d'arbres en utilisant les CRDT ensembles déjà existants. Nous assurons
la cohérence de la structure de données en présence de mutations
concurrentes, en utilisant un algorithme de réparation en une ou deux
phases. Pour chacune de ces phases, nous proposons plusieurs
politiques de réparations indépendantes.  Nous donnons ainsi le choix
au développeur de l'application distribuée le contrôle total sur le
comportement de l'arbre partagé lors de modifications concurrentes.

Enfin, nous proposons d'utiliser des résultats connus et nouveaux sur
les CRDT séquences ordonnés, pour ajoutant des informations de
positionnement sur les noeuds ou les arêtes de l'arbre. Nous
définissons ainsi des CRDT de structure d'arbres ou les noeuds frères
sont ordonnées.}

\RRabstract{Trees are fundamental data structure for many areas of computer
  science and system engineering. In this report, we show how to
  ensure eventual consistency of optimistically replicated trees. In
  optimistic replication, the different replicas of a distributed
  system are allowed to diverge but should eventually reach the same
  value if no more mutations occur. A new method to ensure eventual
  consistency is to design Conflict-free Replicated Data Types (CRDT).

  In this report, we design a collection of tree CRDT using existing
  set CRDTs. The remaining concurrency problems particular to tree
  data structure are resolved using one or two layers of correction
  algorithm. For each of these layer, we propose different and
  independent policies. Any combination of set CRDT and policies can
  be constructed, giving to the distributed application programmer the
  entire control of the behavior of the shared data in face of
  concurrent mutations. We also propose to order these trees by adding
  a positioning layer which is also independent to obtain a collection
  of ordered tree CRDTs.}
  \RRprojets{Concordant}
  \URLorraine

\begin{document}
   \RRNo{7825}
  
  \makeRR

This report is structured as follows. Section~\ref{sec:set} describes
the notion of Conflict-free Replicated Data Types (CRDT). We describe
more precisely the different solutions to build a set CRDT, since all
our tree CRDTs are based on sets. Section~\ref{sec:graph} constructs
several tree CRDTs using the graph theory definition of a tree : a set
of node and a set of oriented edge with some particular properties. To
manage these sets we use set CRDTs; and to ensure the tree properties
in case of concurrent modifications, we build two layers of correction
algorithms. The first layer ensures that the graph is rooted while the
second ensures uniqueness of paths. For each layer, we propose
different and independent policies. Section~\ref{sec:word} also
constructs several tree CRDTs but using word theory to define the set
of paths in a tree. Since such paths are unique, this kind of tree
CRDT is constructed using a set CRDT and a connection layer.
Section~\ref{sec:order} proposes to define ordered tree CRDT by adding
element positioning in tree CRDTs described in previous
sections. These positions come from well-known sequential editing
CRDTs. To make positions compatible with any tree CRDT construct, we
define a new sequential editing CRDT called WOOTR. Finally, we
conclude in Section~\ref{sec:conclusion}.

\section{CRDT  definition}
\label{sec:set}

Replication is a key feature in any large distributed system. When the
replicated data are mutable, the consistency between the replicas must
be ensured. This consistency can be {\em strong} or {\em eventual}. In
the {\em strong consistency} model (aka atomic or linear consistency),
a mutation seems to occurs instantaneously on all replicas. However,
the CAP theorem~\cite{gilbert02brewer} states that is impossible to
achieve simultaneously strong consistency (C), availability (A) and to
tolerate network partition (P).

In the {\em eventual consistency} model, the replicas are allowed to
diverge, but eventually reach the same value if no more mutations
occur. A mechanism to obtain eventual consistency is to design {\em
  Conflict-free Replicated Data Types
  (CRDT)}~\cite{shapiro11conflictfree}.  CRDT can be state-based or
operation-based. In state-based CRDTs -- aka Convergent Replicated
Data Type (CvRDT) -- the data are computed by merging the state of the
local replica with the state of another replica. Eventual consistency
is achieved if the merge relation is a monotonic semilattice. In the
operation-based CRDTs -- aka Commutative Replicated Data Type (CmRDT)
-- the data is computed by executing remote operations on the local
replica. Eventual consistency is achieved if operations are delivered
in certain order and if the execution of the non-ordered operations
commutes. For instance, using causal order, the execution of
concurrent (in Lamport's definition~\cite{lamport78time}) operations
must commutes.

\subsection{Set}

In this section we show how is defined a set CRDT.  We define a data
type by a set of update operations and their pre-condition and
post-conditions. The precondition is local (i.e. it must only be valid
on the replica that generates the update) while the postconditions are
global (i.e. it must be valid immediately after the update).

Consider the operations $add(a)$ and $rmv(a)$ for a set data type. In
a sequential execution, the ``traditional'' definition of the pre- and
post-conditions are
 \begin{itemize}
\item $pre(add(a), S) \equiv a \notin S$
\item $post(add(a), S) \equiv a \in S$
\item $pre(rmv(a), S) \equiv a \in S$   
\item $post(rmv(a), S) \equiv a \notin S$   
\end{itemize}

In case of concurrent updates, the post-conditions $add(a)||rmv(a)$
conflict. Indeed for a CvRDT, we cannot a have a merge that ensure
the both post-conditions. For a CmRDT, the execution of the two
updates in two different orders either leads to two different set
(Figure~\ref{fig:set}), either not ensures the post-conditions.

\begin{figure}[H] 
  \centering
  \includegraphics[width=8cm]{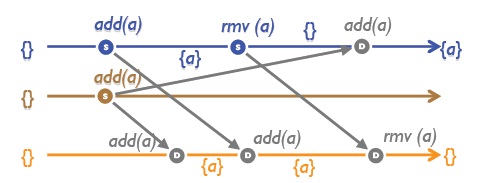}
  \caption{Set with concurrent addition and remove~\cite{shapiro11comprehensive}}
  \label{fig:set}
\end{figure} 

Thus, a set CRDT has different global post-conditions in order to take
into account the concurrent updates while ensuring eventual
consistency. Each CRDT has a {\em payload} which is an internal data
structure not exposed to the client application, and {\em lookup}, a
function on the payload that returns a set to the client
application. For a set CRDT, the pre-conditions must be locally true
on the {\em lookup} of the set.

Different set CRDTs~\cite{shapiro11comprehensive} are the G-Set,
2P-Set, LWW-Set, PN-Set, and OR-Set. They are described below.

\subsection{G-Set}

In a Grow Only Set (G-Set), elements can only be added and not
removed. The CvRDT merge mechanism is a classical set union.

\subsection{2P-Set}

In a Two Phases Set (2P-Set), an element may be added and removed, but
never added again thereafter. The CvRDT 2P-Set (known as
U-Set~\cite{wuu84efficient}) payload consists in two add-only set $A$
and $R$. Adding an element adds it to $A$ and deleting en elements add
it to $R$. The {\em lookup} returns the difference $A \setminus
R$. The set $R$ is often called the tombstones set. 

The CmRDT 2P-Set does not require tombstone but causal delivery; thus,
a remove is always received after the addition of the element.

\subsection{LWW-Set}

In a Last Writer Wins Set (LWW-Set), each element is associated to a
timestamp and a visibility flag.  A local operation adds the element
if not present and updates the timestamp and the visibility flag (true
for $add$, false for $rmv$). The CvRDT merge mechanism makes the union
of all elements and for each element the pair (timestamp, flag) of the
maximum timestamp.

In the CmRDT, the execution of a remote operation updates the element
only if timestamp of the operation is higher than the timestamp
associated to the element. The both CRDTs requires tombstones and the
{\em lookup} returns elements which have a true visibility flag.

\begin{figure}[H] 
\centering
\includegraphics[width=8cm]{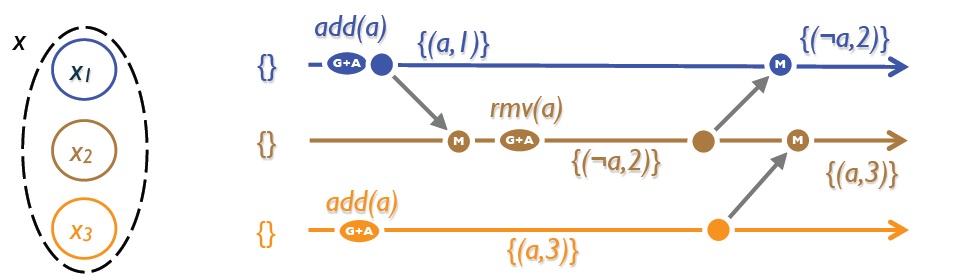} 
\caption{Last Writer Wins Set : LWW-Set~\cite{shapiro11comprehensive}}
\label{fig:LWW} 
\end{figure} 

\subsection{C-Set}

In a Counter Set (C-Set), each element is associated to a counter. Let
$k$ be the value of the counter of an element. A local $add$ can
occurs only if $k\leq 0$ and sets the counter to $1$ ($\delta =
-k+1$). A local $rmv$ can occurs only if $k > 0$ and sets the counter
to $0$ ($\delta = -k$). The CvRDT (also call PN-Set) payload contains
the set of element, and for each element a set $P$ of increments and a
set $N$ of decrements. A local $add$, resp. $rmv$, adds $|\delta|$
element in $P$, resp. $N$. The merge operation is the union of the
sets. The lookup contains elements with $|P|>|N|$.

In the CmRDT, each operation contains the difference $\delta$ obtained
during local execution. The remote operation execution adds $\delta$
to the counter. Element with a counter $k=0$ can be removed, the
others must be kept. The {\em lookup} contains elements with $k>0$.

\begin{figure}[H] 
\centering
\includegraphics[width=8cm]{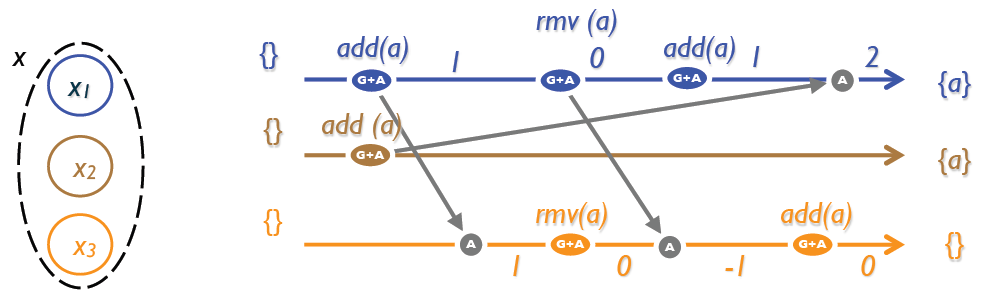} 
\caption{Counter Set : C-Set~\cite{shapiro11comprehensive}}
\label{fig:PN} 
\end{figure}

\subsection{OR-Set}

In a Observed Remove Set (OR-Set) each element is associated to a set
of unique tag. A local $add$ creates a tag for the element and a local
$rmv$ removes all the tag of the element. The CvRDT contains the set of
element, and for each element a set $T$ of tags added and a set $R$ of
tags removed. The merge operation is the union of each set. The lookup
contains elements with $T \cap R\neq \{\}$.

In the CmRDT, each operation contains the tag(s) added or
removed. Since causal ordering is ensured and since tag are unique,
the removed tag (and element with no tag) can be removed in the
payload. The {\em lookup} contains the elements of the payload.

\begin{figure}[H] 
\centering
\includegraphics[width=8cm]{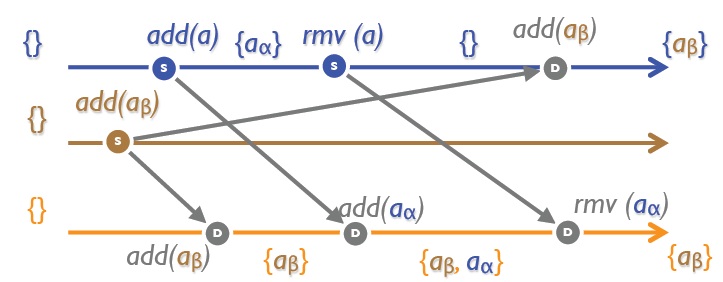} 
\caption{Observed-Remove Set : OR-Set~\cite{shapiro11comprehensive}}
\label{fig:OR} 
\end{figure} 

\subsection{Comparison}

From the application point of view, all set CRDTs provide a set lookup
and the same pre-conditions on operations (except for G-Set, since the
application cannot remove an element and for 2P-Set, since application
cannot re-add an element). They also provide the same post-conditions
of the local replica. The behavior of the presence of the elements in
the lookup can be resumed as follow~:
\begin{description}
\item[LWW-Set] an element appears in the lookup if and only if the
  operation with the higher timestamp is an $add$.
\item[C-Set] an element appears in the lookup if and only if the
  sum of the $add$ operations counters is greater than the sum of the
  $rmv$ operations counters.
\item[OR-Set] an element appears in the lookup if and only if the
  tags associated by $add$ operations are not all present in $rmv$
  operations.
\end{description}

\section{Graph Trees}
\label{sec:graph}

According to standard graph theory definition, a tree -- more precisely
an arborescence -- is a connected directed acyclic graph in which a
single node $root$ is designated as the root and there is a
unique path from $root$ to any other node \cite{diestel2010graph}. A
tree is thus a ordered pair $G=(V,E)$ with $V$ a set of nodes and
$E\subseteq V\times V$ a set of directed edges. If $(x,y) \in V$, we
say that $y$ is a {\em child} of $x$. Since $G$ have no directed
cycle, $E^{*}$, the transitive closure of $E$, is a partial strict
order on $V$. There is a path from $x$ to $y$ if and only if $(x,y)
\in E^{*}$.

We define {\em subtrees} in a more general manner than usual by
including edges directed to the subtree. In an actual tree there is
only one such edge.
\begin{definition}
  An ordered pair $(N, F)$ is a subtree of the tree $(V,E)$ and is
  rooted by $n\in N$ if $N \subset V$, $F \subseteq E$ and $(N, F
  \setminus ((V \setminus N) \times N))$ is a connected directed
  acyclic graph with a unique path from $n$ to any other node and if
  $(V \setminus N, E \setminus F)$ is a tree.
\end{definition}

We consider that the graph can be modified trough two minimal
operations $add$ and $rmv$. The operation $add(n,m)$ adds a node $n$
in the graph under the node $m$ and the operation $rmv(N, F)$ removes
the set of nodes and edges appearing in a subtree. Other 
operations, e.g. adding a whole subtree, or removing a node while
moving all its children under the father of $n$, can be defined upon
these minimal operations\footnote{For instance, adding a whole subtree
  consists of a list of $add$ operations; remove a node while
  keeping its children consists of a list of $rmv$ and a list of
  $add$.}.
We have the following formal definition of the sequential operations
on a tree. For sake of simplicity, we consider that the root of the
tree is always present and immutable.
\begin{itemize}
\item $pre(add(n,m), (V,E)) \equiv n \notin V \wedge m \in V$
\item $post(add(n,m), (V,E)) \equiv n \in V \wedge (m, n) \in E$
\item $pre(rmv(N, F), (V,E)) \equiv subtree((N,F),(V,E))$   
\item $post(rmv(N, F), (V,E)) \equiv N \cap V = \{\} \wedge F \cap E = \{\}$.
\end{itemize}

With such pre- and post-conditions we can ensure that the graph
$(V,E)$ stays a tree in case of a sequential modifications. However,
in case of a concurrent modifications, these post-conditions conflicts
if a node is concurrently added and removed, if a node is concurrently
deleted while a children is added, or if a node concurrently added
under to different fathers.

\subsection{Concurrent addition and deletion of the same element}

The post-conditions of $add(n,m)||rmv(N, F)$ with $n \in N$
conflicts, i.e. a node cannot be concurrently added and
removed. Indeed, as for a set, the post-condition of $add$ and $rmv$
operations cannot be globally ensured while ensuring convergence.

We can uses sets CRDT to bypass the conflict.  By using sets CRDT to
handle both sets of nodes and edges, we obtain a data type $(V,E)$
that is obviously eventually consistent. Such trees CRDT have the
following behavior.

\begin{description}
\item[GG-Tree] In a Grow-only Graph Tree (GG-Tree) nodes and edges can
  only be added and never removed. A GG-Tree uses G-Sets as the sets
  of nodes $V$ and edges $E$.
\item[2G-Tree] In a Two-phases Graph Tree (2G-Tree) nodes and thus
  edges can only be added once. A 2G-Tree uses the lookup of a 2P-set
  as the set of nodes $V$.  There is no need for using set CRDT for
  the edges since a new edge is only added with a new node. Thus, an
  edge cannot be added and removed concurrently.
\item[LG-Tree] In a Last-writer-wins Graph Tree (LG-Tree) a node, or a
  edge, appears in the lookup if and only if the operation with the
  higher timestamp applied on it is an add. LG-Tree uses the lookup of
  LWW-element-Sets as the sets of nodes $V$ and edges $E$. The
  operations become $add(n,m,t)$ and $rmv(N,F,t)$. The execution of
  the operations consists in updating the timestamp and the visibility
  flag if the operation timestamp is newer that the attached
  timestamp.
\item[CG-Tree] An Counter Graph Tree (CG-Tree) a node or a edge
  appears in the lookup if and only if the sum of $add$ operations
  applied on it is greater than the sum of $rmv$ operations. A CG-Tree
  uses the lookup of C-Sets as the sets of nodes $V$ and edges
  $E$. The operation $add$ and $rmv$ associate an increment to each
  element appearing in these operation.  The execution of the
  operation applies this positive or negative increment to the targeted
  elements.
\item[OG-Tree] In an Observed-remove Graph Tree (OG-Tree), a node or a
  edge appears in the lookup if and only if the tags associated by
  $add$ operations applied on it are not all removed by $rmv$
  operations. An OG-Tree uses the lookup of OR-Sets as the sets of
  nodes $V$ and edges $E$.  The operation $add$ associates a unique tag
  and $rmv$ associates a set of tag to each element appearing in these
  operation.  The execution of the operation adds or removes the
  tag(s) to the targeted elements.
\end{description}

\subsubsection{Set lookup}

From all the above tree CRDTs, we can obtain $(V_L, E_L)$ a pair of
lookup sets which is eventually consistent since lookup of the set
CRDTs is eventually consistent. However, in case of concurrent
modifications, this pair $(V_L, E_L)$ is not a graph since $E_L$ may
contain edge on nodes not in $V_L$. For instance in the LG-Tree, if
the operations $add(n, m, t)$ and $rmv(N, F, t')$ with $n\in N$ and
$t'>t$ are generated concurrently, we get $(m,n) \in E_L$ while $n
\notin V_L$.

The pair $(V_L, E_L \cap (V_L \times V_L))$ is a graph but may not be
a tree. It can be non-connected if a replica adds a node under $m$ and
another replica removes $m$ concurrently. Also, there can be several
paths between the root and a node since two replicas can add
concurrently a node under two different fathers. Moreover, such a
graph may contains cycles if, for instance, a replica generates
$add(x,root)$ followed by $add(y,x)$ and another replica generates
concurrently $add(y,root)$ followed by $add(x,y)$.

\begin{figure}[H] 
\centering
\includegraphics[width=6cm]{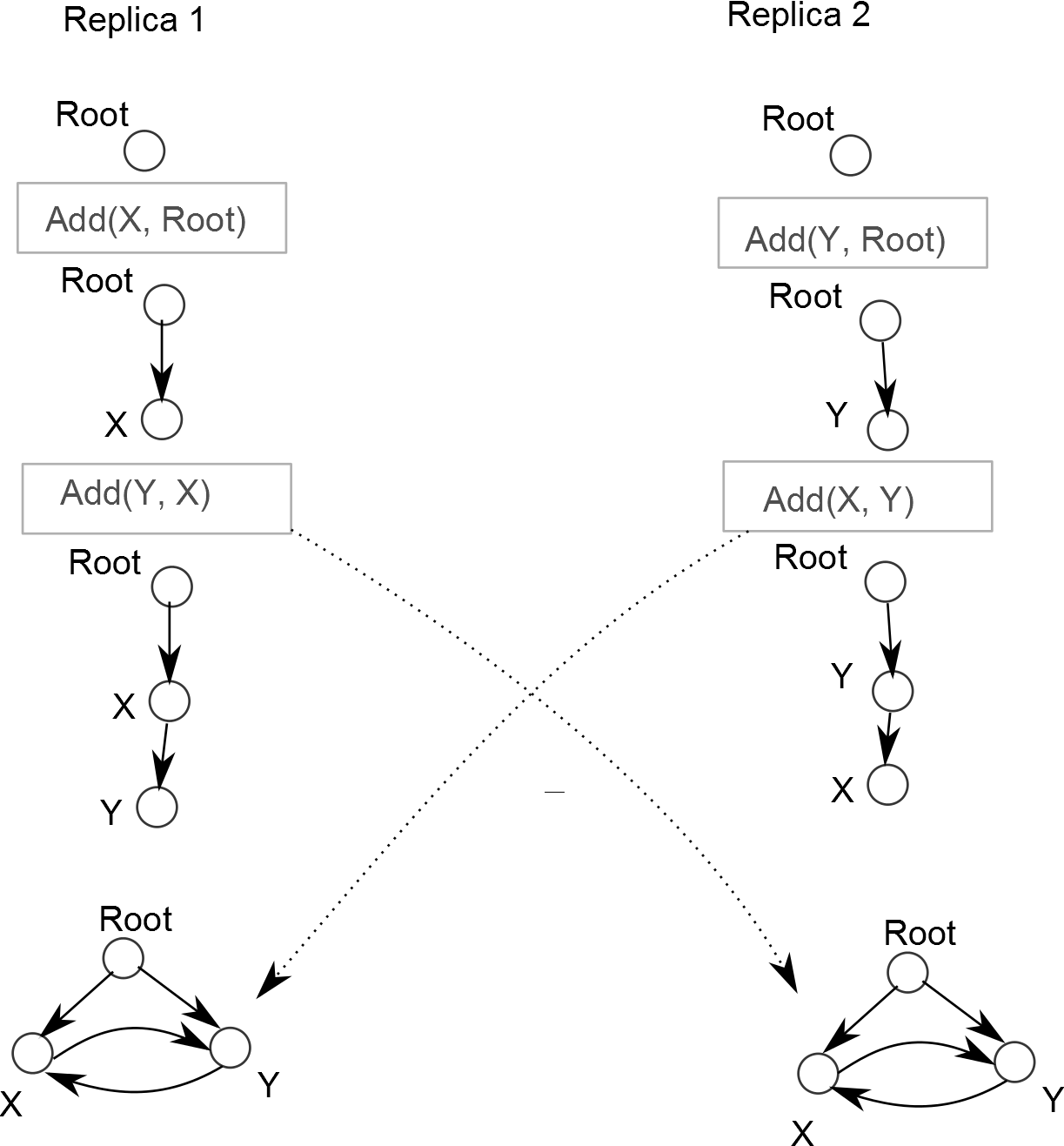} 
\caption{Cycle generated by concurrent additions}
\label{fig:cycles} 
\end{figure} 

We propose to compute a lookup from the pair $(V_L, E_L)$ on order to
obtain a lookup which is an eventually consistent tree. In the
following sections, we propose different policies to firstly reconnect
or drop the isolated components to obtain a rooted graph, and to
secondly to express a tree from the rooted graph.

\subsection{Connection policy}
\label{glue}

The operations $add(n,m)||rmv(N, F)$ with $m \in N$ and $n \notin N$
conflicts since a naive lookup of the underlying sets CRDTs of nodes and
edges is a non connected graph.  However, several solutions can be
designed to produce a graph which is rooted, i.e. with at least one
path from the root to any other node. The solutions can be to ``skip''
such $add$, to ``recreate'' the removed ascendant(s), or to place such
added nodes ``somewhere'' in the tree (for instance under the
root). We compute a rooted graph $(V_C, E_C)$ directly from the
lookup $V_L$ and $E_L$ of the supporting sets CRDTs.

We note $E_G=(E_L\cap (V_L\times V_L))$. We call a {\em orphan node},
a node $n$ in $V_L$ such that $(root,n) \notin E_G^{*}$. Since a node
is always added with an edge directed to it, an orphan node $n$ has at
least one edge in $(m,n) \in E_L$ directed to it; if $m\notin V_L$, we
call $(m, n)$ {\em an orphan edge}, elsewhere $m$ and $n$ are parts of
the same {\em orphan component}.

To compute $(V_C, E_C)$, we start by adding all non-orphan nodes and
the edges between them in $(V_L, E_L)$. Then, we treat the orphan
nodes in $V_L$. Considering each orphan node $n$, we can apply the
following {\em ``connection'' policies}~:

\begin{description}
\item[skip] {\em drops the orphan node.} This algorithm consists
  simply on a graph traversal starting from the root and is in
  $\Theta(|E_L|+|V_L|)$.

\begin{figure}[H] 
\centering
\includegraphics[width=9cm]{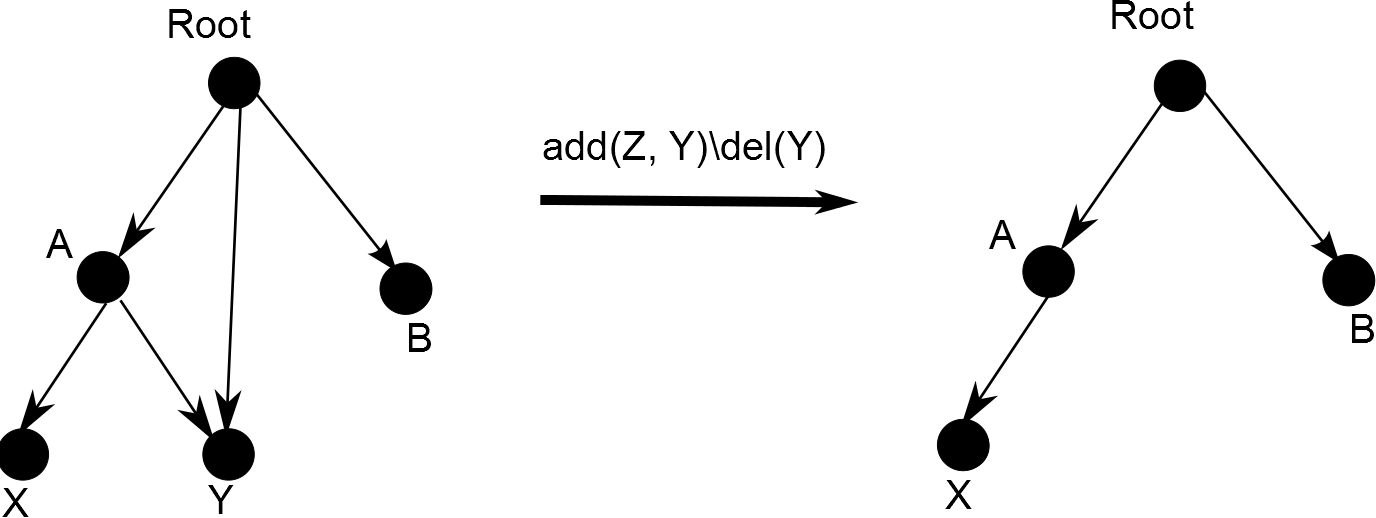} 
\caption{skip policy}
\label{fig:Skip} 
\end{figure}

\item[reappear] {\em recreates all paths leading to orphans
    components}. We add all edges $(n, y)$ such that $y \in V_L$.  For
  each orphan edge $(x, n)$ we add all paths (nodes and edges) that
  have ever existed between $root$ and $y$.  This policy requires to
  keep as tombstones all the edge ever added to the graph.  This
  algorithm is in $\Theta(|E|+|V|)$.

\begin{figure}[H] 
\centering
\includegraphics[width=9cm]{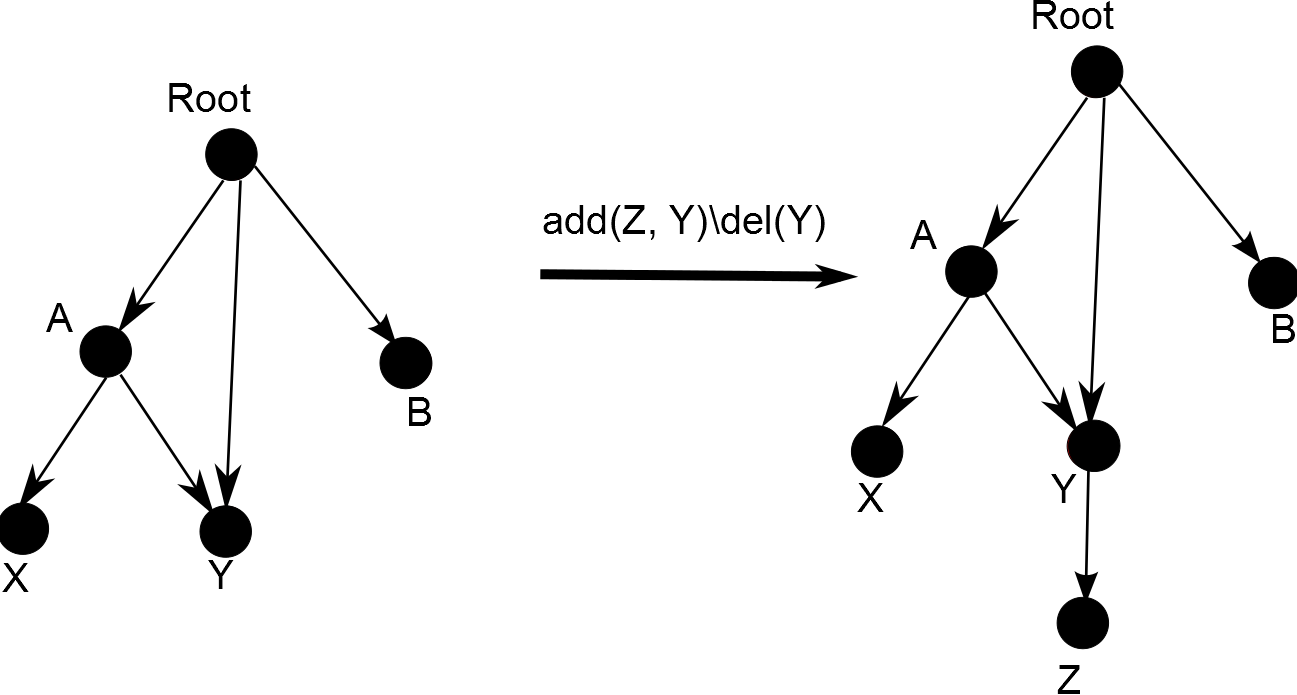} 
\caption{reappear policy}
\label{fig:Reappear} 
\end{figure}

%


\item[root] {\em places the orphan components under the root.} We add
  all edges $(n, y)$ such that $y \in V_L$. For each orphan edge $(x,
  n)$, we add $(root,n)$.  This algorithm consist in modification of
  all orphans edges : we replace inexistent node by root.  This
  algorithm complexity is $\Theta(|E_L|+|V_L|)$.

\begin{figure}[H] 
\centering
\includegraphics[width=9cm]{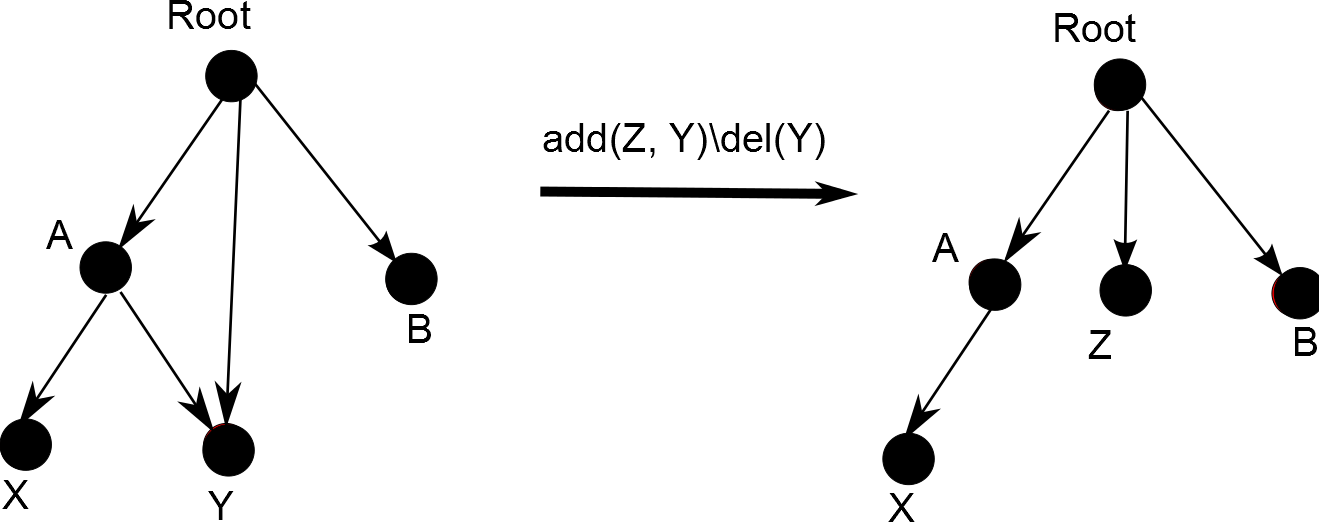} 
\caption{root policy}
\label{fig:Root} 
\end{figure}

\item[compact] {\em places the orphan components under the connected
    node that have ever a path to it}. We add all edges $(n, y)$
  such that $y \in V_L$. For each orphan edge $(x, n)$, we add $(z,
  n)$ for all $z$ which is a non-orphan node such that a path that does
  not contains non-orphan nodes have ever existed between $z$ and $x$.
  This policy requires to keep as tombstones all the edge ever added
  to the graph.  Let $connectSet$ be a set associated on all node. By default
  this set is empty. 
  We execute the follow
  algorithm with node previously deleted and connected to orphans edges.

 \begin{lstlisting}
function  getConnected(node n)
   if n is not orphan then
         return {n}
    endif
    if n.connectSet is empty then
           for n' in father node
                 n.connectSet.add(getConnected(n'))
           return n.nonnectSet;
    else
           return n.connectSet
    end if
  \end{lstlisting}
 
 \begin{figure}
 \centering
 \includegraphics[width=9cm]{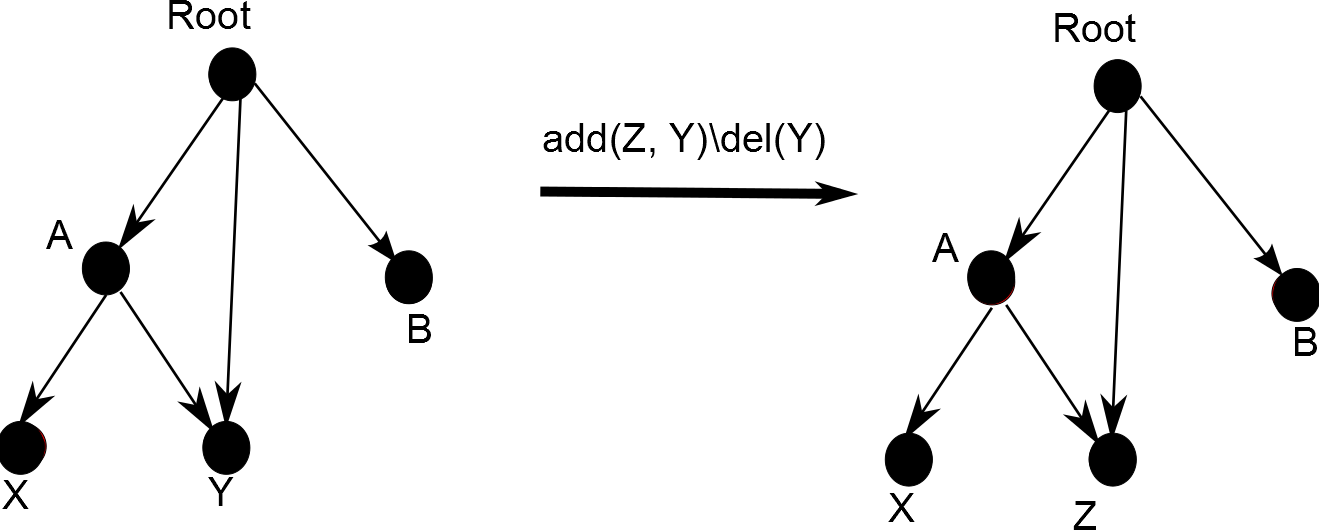}
 \caption{Compact policy}
\label{fig:Compact} 

 \end{figure}
 Finally for all orphan edge we add all edges link connected from each element returned by algorithm to the component node.
 This algorithm is in $\Theta(|E|+|V|)$
\end{description}

Using any of the above policies ensures that $(V_C, E_C)$ is a rooted
graph for any tree CRDT. Such a graph is eventually consistent and
there is at least one path from $root$ to each other node.

The {\em reappear} and {\em compact} policies require to keep all edges
that have ever existed as tombstones. In some set CRDTs approaches
(such as CmRDT LWW-Set or all set CvRDTs), these tombstones already
exist in the payload. In the {\em compact} policy, we can store only
the set of node that have ever been accessible from one node.

 
\subsection{Mapping policy}
\label{depolyhandrycyclotron}

The operations $add(n,m)||add(o,p)$ with $n=o$ and $m\neq k$
conflict. A node cannot be concurrently added under two different
nodes, since the graph may contains different paths to a node and
directed cycles. To obtain a tree we start from the rooted graph
$(V_C, E_C)$ and we apply one of the three following {\em
  ``mapping'' policies}.

\begin{description}
\item[several] : We construct all the acyclic paths in the
  graph. Thus, copies of the node can appear in different places in
  the tree. Remove a copy of the node removes all the others. The
  algorithm is a simple depth-first that begins on root node. For each
  node, the algorithm is
  \begin{enumerate}
  \item Mark the node.
  \item Construct a list $l$ composed of recursive calls on all
    unmarked children nodes.
  \item Unmark the node.
  \item Returns a tree composed of the node and the list $l$ of
    children
  \end{enumerate}
  Obtaining a description of all simple paths in a directed graph can
  be computed using $O(|V|^3)$ matrix operations. Such a tree contains
  up to $|V|!$ edges in case of a complete graph.
  
\item[one] : This policy adds in the tree each node in $V_L$ only once. Thus,
  the algorithm must make a choice on the edges~:
  \begin{description}
  \item[newer] : The ``newer'' variation needs a timestamps on edges
    to select the newer. This is adapted to LG-Tree that already has
    such timestamp\footnote{This is also adapted to OG-Tree if tag are
      constructed with clocks}. We construct a Maximal Spanning Tree
    (MST) with the edges sorted by timestamp. We will not obtain a
    tree composed with only newest edges since such edges may
    constitute a cycle. But we will obtain a tree with the maximal sum
    of timestamp. This tree will be rooted since $root$ has no edge
    directed to it and must be included in the MST. Building a MST in
    a directed graph can be achieved in $\Theta(|E| + |V| log
    |V|)$~\cite{gabow86efficient}.
  \item[higher] : This variation is designed for CG-Tree and OG-Tree. We
    construct a MST maximising edges counters or edge tags numbers.
  \item[shortest] : This variation can be used for all type of
    tree. For each node we select the shortest path to it. A
    Breath-first algorithm can be used to produce the tree in
    $\Theta(|E| + |V|)$.
  \end{description}
  
\item [zero] : The zero policy removes all subtrees rooted by nodes
  which have more than one edge directed to them. For each node the
  algorithm checks the number of input edges. The algorithm traverses
  the graph starting from the root but does not add nodes with an in
  degree greater than two and does not visit its children. The
  algorithm is in $\Theta(|V|+|E|)$.
\end{description}


\subsection{Discussion on graph trees}

Thus, we can obtain a lookup using a graph structure managed by set
CRDTs. This lookup function is composed in three phases. The first
phase is the lookup of the underlying set CRDT. The second phase
computes a rooted graph. The third phase expresses a tree from the
rooted graph. Such data types are obviously CRDTs since the underlying
sets are eventually consistent, and since the lookup tree is computed
with deterministic policies\footnote{We assume existence of a total
  order between nodes to ensure determinism of graph algorithms.},
this lookup is also eventually consistent.

However, depending on the policy chosen, the client application can
observe moves on the lookup tree. For instance, using a root policy,
if a removed father is added again, its orphan son will move from the
root to its original place. 

We call {\em monotonic policy}, a policy where $add$ and $rmv$
operations do not move an existing node in the lookup. The
non-monotonic policies are : root, compact and all one variations. The
monotonic policies are skip, reappear, zero and several.

The lookup function works after each modification of the tree.  The
complexity of this function could be up to factorial for the several
policy. So, some optimizations are useful.  We call {\em incremental}
lookup function, a lookup function which reuses an previous calculus
to avoid recompute entire tree.  For example, in the reappear policy,
when an orphan node should be added, the incremental lookup function
adds to the tree the several/one/zero paths leading to this node.  On
the other hand, when the father of an orphan node is added, the other
reappeared paths must be removed of the lookup. Finally, when an
orphan node is removed, the reappeared paths should disappear. Such
incremental versions have the same worst-case complexity than
non-incremental ones but are slightly more efficient. However,
eventual consistency of the lookup is less straightforward to ensure
in such incremental versions.



\subsection{A special case : 2G-Tree} 
\label{sec:2G}

A two phases graph tree (2G-Tree) uses a
2P-set~\cite{shapiro11comprehensive} as the set of nodes $V$. A 2P-Set
consists in defining unique elements that can only be added once on
all replicas.  Thus, node and edge cannot be added and removed
concurrently.  The other main advantage of the 2G-Tree is that the
conflict $add||add$ does not occurs since a node can only be added
once. Thus, 2G-Tree do not require any mapping policy.

In a 2G-Tree, the conflict $add(n,m)||rmv(N, F)$ with $m \in N$ and $n
\notin N$ can be resolved using solutions presented above. Assuming
that node can be found in constant time (using hash table), the skip
policy can be computed incrementally in $\Theta(1)$ time. Indeed, the
remove of a node consists in remove of the entire subtree, and
addition of an orphan node has no effect. Moreover a CvRDT 2G-Tree can
send constant size messages for remove : $rmv(n)$ with $n$ the root of
the subtree.  The reappear and compact policy can be computed in
$\Theta(|V|)$ since there is only one path, of size at most $|V|$,
leading to a given node. Finally, in the root policy, the addition of
a node is always in $\Theta(1)$ time.

\subsection{EDGE trees}

Since a node is always added with an edge directed to it, one can
represent a tree using only edges. Such a choice leads to a data
structure we call {\em edge tree}. Given a finite or infinite set of
nodes $V$, an edge tree is a subset of all ordered pairs. An edge tree
has a root with no edge directed to it, and for all edge, it exists one
unique parent edge. A subtree is rooted by a node and include the
edge\footnote{In case of concurrent modifications, their can be
  several such edges.} directed to this node and a set of connected
edges.

\begin{definition}
  An {\em edge tree} $T$ rooted by $root$ is a subset of $V\times V$
  such that for all $(x,y) \in T$ either $x = root$, or there exists a
  unique $z \in V$ such that $(z,x) \in T$.

  The set $S$ is a {\em subtree} rooted by $n\in V$ of $T$ if $S
  \subset T$, $\exists (x, n) \in S$, $\forall (a, b) \in S.~b \neq n
  \implies (n, b) \in S^{*}$ and $T \setminus S$ is an edge tree.
\end{definition}

We have the following formal definition of the sequential operations
on an edge tree.
\begin{itemize}
\item $pre(add(n,m), T) \equiv \exists (z,m) \in T$
\item $post(add(n,m), T) \equiv (m, n) \in T$
\item $pre(rmv(S), T) \equiv subtree(S,T)$   
\item $post(rmv(S), T) \equiv S \cap T = \{\}$.
\end{itemize}

As for graph tree, the post-conditions of $add$ and $rmv$ conflict and
an edge tree CRDT uses a set CRDT to handle the set of edges.  We can
apply the same connecting and mapping policies than for graph tree to
compute a {\em tree lookup} of the CRDT set. We simply consider that a
node belong to a tree if and only if it appears on an edge of tree.

Such GE-Tree, 2E-Tree or OE-Tree will have exactly the same behavior
than respectively GG-Tree, 2G-Tree and OG-Tree. Indeed, in such trees,
we cannot remove edges (GG-Tree), or we cannot have an edge directed
to a removed node (2G-Tree and OG-Tree). Thus, GE-Tree, 2E-Tree and
OE-Tree are optimizations of their respective xG-Tree.

The LE-Tree and CE-Tree have a different behavior than LG-Tree and
CG-Tree. Indeed, let consider a first replica that inserts a node $x$
under a node $y$, and then removes $x$, while a second replica insert
$x$ under a node $z$. Depending on the timestamps (LG-Tree) or on if
another replica removes $(y, x)$ concurrently (CN-Tree), the node $x$ --
and thus $(z, x)$ -- can appear or not in the lookup. In LE-Tree and
CE-Tree, $(z, x)$ appears in the lookup.

\section{Word trees}
\label{sec:word}

In this section we introduce {\em word trees}, another data structure
to manage concurrently modified trees. A word represents a path in the
tree, a tree can be defined as a set of words~: the set of paths
existing in this tree. We use the standard definitions about words.
Let $\Sa$ be a finite -- or infinite -- ordered alphabet, a word is a
finite sequence of elements from $\Sa$. The length of a word $w$,
noted $|w|$ is the number of elements of $w$. We denotes $\e$ the
empty word. The concatenation $vw$ is the word formed by the joining
end-to-end the words $v$ and $w$. The set of all strings over $\Sa$ of
any length is the Kleene closure of $\Sa$ and is denoted $\Sk$.

We define a tree as a set of the words representing all the paths in
the tree. Since all the paths are present in the set, any prefix of a
path is also a path of the tree. The empty word $\e$ is the root of
the tree.
\begin{definition}
  A word tree $T$ is a subset of $\Sk$, such that $\e \in T$ and
  $\forall p,e \in \Sk.~pe\in T \implies p\in T$.
\end{definition}

A subtree is defined as complete set of paths with a common prefix.
\begin{definition}
  In a tree $T$, a subtree $P$ is a subset of $T$ such that $T
  \setminus P$ is a tree and such that $\exists w \in T.~\exists S
  \subset \Sk.~P=\{ws|s\in S\}$ and $S$ is a tree.
\end{definition}

As for graph tree, there is two operations to modify a word tree. The
operation $add(n,p)$ with $n\in\Sa$ and $p\in\Sk$ adds a new path and
$rmv(P)$ removes the set of paths representing a subtree.
\begin{itemize}
\item $pre(add(n, p), T) \equiv p \in T \wedge pn \notin T$
\item $post(add(n, p), T) \equiv pn \in T$
\item $pre(rmv(P), T) \equiv P \subset T \wedge subtree(P, T)$   
\item $post(rmv(P), T) \equiv P \cap T = \{\}$
\end{itemize}

With such pre- and post-conditions, we can ensure that the set $T$ is
sill a tree in case of sequential modifications.  In case of
concurrent modifications, word trees differ from graph trees since
only $add||del$ conflicts occurs.

\subsection{Concurrent addition and remove of the same element}

A for mathematical set, the post-conditions of $add(n,p)$ and $rmv(P)$
with $np \in P$ conflicts since convergence cannot be achieved. As for
graph trees, we can use set CRDT to bypass the conflict. The obtained
tree CRDT have the following behavior~:

\begin{description}
\item[GW-Tree] a path can only be added and never removed.
\item[2W-Tree] a path can only be added once. Such a CRDT has the same
  behavior than the 2G-Tree and 2E-Tree.
\item[LW-Tree] a path appears in the lookup if and only if the
  operation with the higher timestamp applied on it is an $add$.
\item[CW-Tree] a path appears in the lookup if and only if the number
  of $add$ operations applied on it is greater than the number of
  $rmv$ operations.
\item[OW-Tree] a path appears in the lookup if and only if the tags
  associated by $add$ operations applied on it are not all removed by
  $rmv$ operations.
\end{description}

All the above data types are obviously eventually consistent. But the
lookup presented must be a tree even in case of the concurrent
addition of a node and remove of its father.

\subsection{Concurrent addition of a path and remove of the prefix}

As for graph and edge trees, the naive execution of operations
$add(n,p)$ and $rmv(P)$ with $p \in P$ produce a set of path which is
no longer a tree. Thus we need to compute a lookup which is a tree. We
compute a lookup tree $LT$ from the set of path $LS$ obtained from the
lookup of the supporting set CRDT.

We call a orphan path, a path in $LS$ that has a prefix which is not
in $LS$. We start by adding all non-orphan paths of $LS$ to
$LT$. Then, we treat the orphan paths in $LS$ in length order
(shortest first, then $\Sa$ order). Considering each orphan
path $a_1a_2\ldots a_n \in LS$ with $\forall i \in [1,n].~a_i\in \Sa$,
we can apply the following {\em connection policies}~:

\begin{description}
\item[skip] {\em drops the orphan path.}
\item[reappear] {\em recreates the path leading to the orphan path.}
  We add all $a_1 \ldots a_j$ with $j \in [1,n]$.
\item[root] {\em places the orphan subtree under the root.} We add $a_j
  \ldots a_n$ to $LT$ with $j$ such that $a_1 \ldots a_{j-1} \notin
  LS$ and $\forall k \in [j,n]$, $a_1 \ldots a_{k} \in LS$.
\item[compact] {\em places the orphan subtree under its longest
    non-orphan prefix.} We add $a_1 \ldots a_ma_j \ldots a_n$ to $LT$
  with $j$ and $m$ such that $m<j$ and $a_1 \ldots a_m \in LT$ and
  $a_1 \ldots a_{m+1} \notin LS$ and $a_1 \ldots a_{j-1} \notin LS$
  and $\forall k \in [j,n]$, $a_1 \ldots a_{k} \in LS$.
\end{description}

\begin{example}
  For a lookup $LS=\{\e, a, ab, ac, abcd, abcde, abcdefg\}$, the
  orphans path are $\{abcd, abcde, abcdefg\}$ and we obtain $LT$ equal
  to~:
\begin{description}
\item[skip]     $\{\e, a, ab, ac\}$
\item[reappear] $\{\e, a, ab, ac, abc, abcd, abcde, abcef, abcdefg\}$
\item[root]     $\{\e, a, ab, ac, d, de, g\}$
\item[compact]  $\{\e, a, ab, ac, abd, abde, abdeg\}$
\end{description}
\end{example} 

Using any of the above policies ensures that the lookup trees presented
to the client by any CRDT tree are eventually consistent.

\begin{theorem}
  The lookup sets $LT$ computed using a skip, root, reappear, or compact
  policy are tree and are eventually consistent.
\end{theorem}

\begin{proof}
  Since the set of paths $LS$ is eventually consistent, and since the
  paths are treated is the same order and since each policy is
  deterministic, the computed set of paths $LT$ is eventually
  consistent.

  Set of path $LT$ is a tree since~:
\begin{description}
\item[skip] there is no orphan path in $LT$.
\item[reappear] we add an orphan path in $LT$ with all its prefixes.
\item[root] a suffix $a_j \ldots a_n$ is added to $LT$ only if $\forall
  k \in [j,n]$, $a_1 \ldots a_{k} \in LS$. Thus, all the prefixes $a_j
  \ldots a_k$ were also added to $LT$.
\item[compact] a path $a_1 \ldots a_ma_j \ldots a_n$ is added to $LT$ only
  if $\forall k \in [j,n]$, $a_1 \ldots a_{k} \in LS$. Thus all the
  prefixes $a_1 \ldots a_ma_j \ldots a_k$ were also added to $LT$.
\end{description}
\end{proof}

Computing a lookup tree $LT$ every time the lookup set $LS$ is
modified ensures easily eventual consistency, but only some policies
are monotonic. We consider a policy as {\em monotonic} if the $add(p)$
operation do not moves any already existing node in tree. The {\em
  root} and {\em compact} policies are not monotonic since when the
missing ascendants are added again, the orphan subtree moves to its
original place.

The advantage of monotonic policies is that the client of the tree
CRDT will not observe such move, and that a client operation on an
orphan path do not require a complex translation into an operation on
the supporting set CRDT.

\subsection{Complexity and optimisation}

Lets assume that a hash table is used to implement the set of
paths. Thus, checking for all prefixes of path if they belongs to a set
have an average time complexity proportional to the length of the
path. Thus, the time complexity to apply a policy to a path is linear.
Also, the time complexity to compute a lookup tree is $\Theta(pk)$ in
average, with $p$ the number of paths in $LS$ and $k$ the average
length of these paths. The worst case time complexity is $\Theta(n^2)$
with $n$ the number of paths in $T$.

However, at least for the monotonic policies, we can compute $LT$
incrementally, i.e. without parsing the whole set $LS$.

\begin{description}
\item[skip] When an orphan path is supposed to be added in the lookup,
  we drop it. When an non-orphan path $p$ is added, we add recursively
  all $pa \in LS$ with $a \in \Sa$. When a path $n$ is supposed to be
  removed in the lookup, we remove all the paths that are prefixed by
  $n$. Moreover, a tree CmRDT can send only the operation $rmv(n)$
  with $n$ the common prefix of the subtree, since the whole subtree
  will be removed.
\item[reappear] In the {\em reappear} policy, when an orphan path is
  removed we must remove the reappeared path to ensure eventual
  consistency.  This can be done by marking the reappeared paths as
  ``ghosts''.  When path previously marked as ghost is supposed to be
  added in the lookup, we unmark it. When an orphan path $n$ is
  supposed to be added in the lookup, we add all the prefixes of $n$
  that are not existing and we mark them as ghost. When a path $n$ is
  supposed to be removed in the lookup, if $n$ is the prefix of a
  non-ghost path, we mark $n$ as ghost, elsewhere we remove $n$ and
  all the ghost prefixes of $n$ that are the prefixes of not any ghost.
\end{description}


\section{Ordered trees}
\label{sec:order}

In this section, we present ordered trees, where the set of children
of a node is totally ordered. For this we need to add to the unordered
trees presented above, an additional information called {\em Position
  Identifier (PI)} which allows to order the children. These position
identifiers must be totally ordered to ensure eventual consistency and
defined within a dense space to allow insertion of a node at an
arbitrary position. These position identifiers can be associated to
nodes or edges.

To obtain position identifiers, an idea to use PI already defined for
sequence editing CRDTs such as Logoot~\cite{weiss10logootundo},
Woot~\cite{oster06data}, WOOTO~\cite{weiss07wooki},
RGA~\cite{roh11replicated} or
Treedoc~\cite{preguica09commutative}. Such PIs are {\em Unique
  Position Identifier (UPI)} and thus constrain the behavior of the
trees to some kind of two-phases set that does not allow concurrent
insertions of the same element or re-insertions. So, we propose a new
non-unique position identifier to allow such operations.


In the following figures, a plain arrow represents the child relation
between node, and a dotted arrow represents the order between
children.

\subsection{Unique positioning for nodes}
 
We associate each node to an unique position identifier (UPI). The
order between the children of a node is given by their UPI.  Since
only graph trees manage nodes and since position identifiers are
unique, we obtain a 2G-Tree. If a node is added twice concurrently,
even at the same place in the ordered tree, we obtain two different
nodes. The formal definition of the operation $rmv$ do not change, a
node is a pair $(element, UPI)$ and an edge is a pair of node. The formal
definition of operation \textit{add} becomes :
\begin{itemize}
\item $pre(add((n,u),(m,v)), (V,E)) \equiv (n,u) \notin V
  \wedge (m,v) \in V \wedge unique(u)$
\item $post(add((n,u),(m,v)), (V,E)) \equiv (n,u) \in V \wedge
  ((m,v), (n,u)) \in E$
\end{itemize}

The conflict $add||add$ does not occurs, since a node can only be
added once with an UPI. In figure \ref{fig:Diff_node}, a replica
produces $add(Z,A)$ while another replica produces $add(Z',B)$
concurrently, but they are considered as two different elements with
same characteristics.

\begin{figure}[H] 
\centering
\includegraphics[width=10cm]{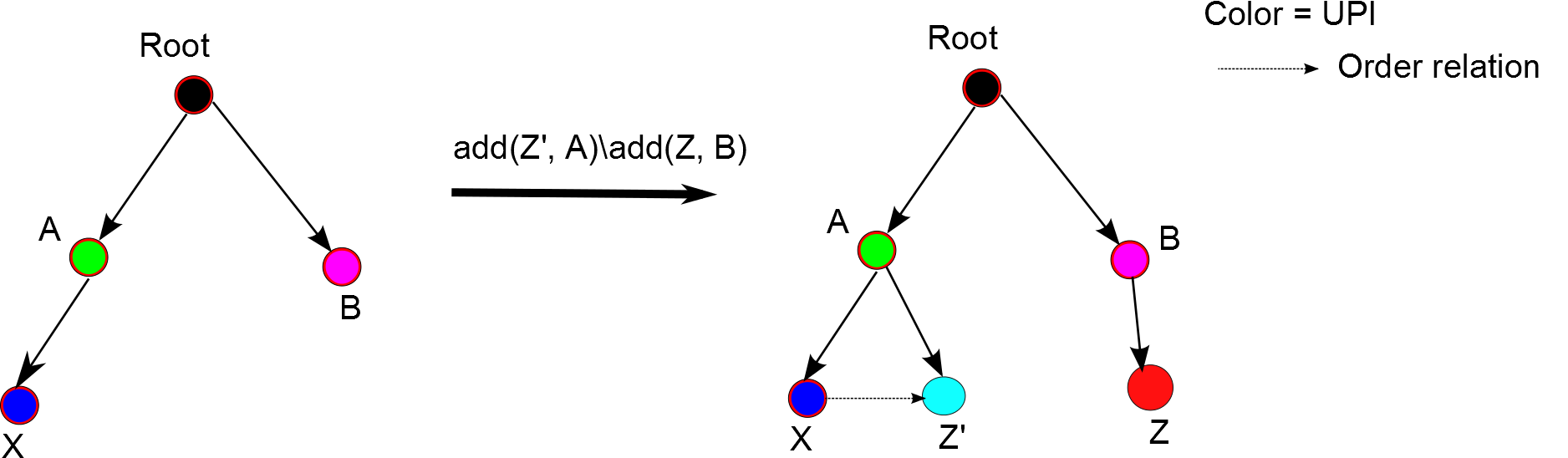} 
\caption{Concurrent operations add/add with node positioning}
\label{fig:Diff_node} 
\end{figure}

The conflict $add(n,m)||rmv(N, F)$ with $m \in N$ and $n \notin N$ can
be resolved with the same policies defined for 2G-Tree in
Section~\ref{sec:2G}. In figure~\ref{fig:Figure_add_del}, we represent
the execution of two concurrent operations $add(Z,Y)/rmv(Y)$ with
the {\em skip} policy.

\begin{figure}[H] 
\centering
\includegraphics[width=9cm]{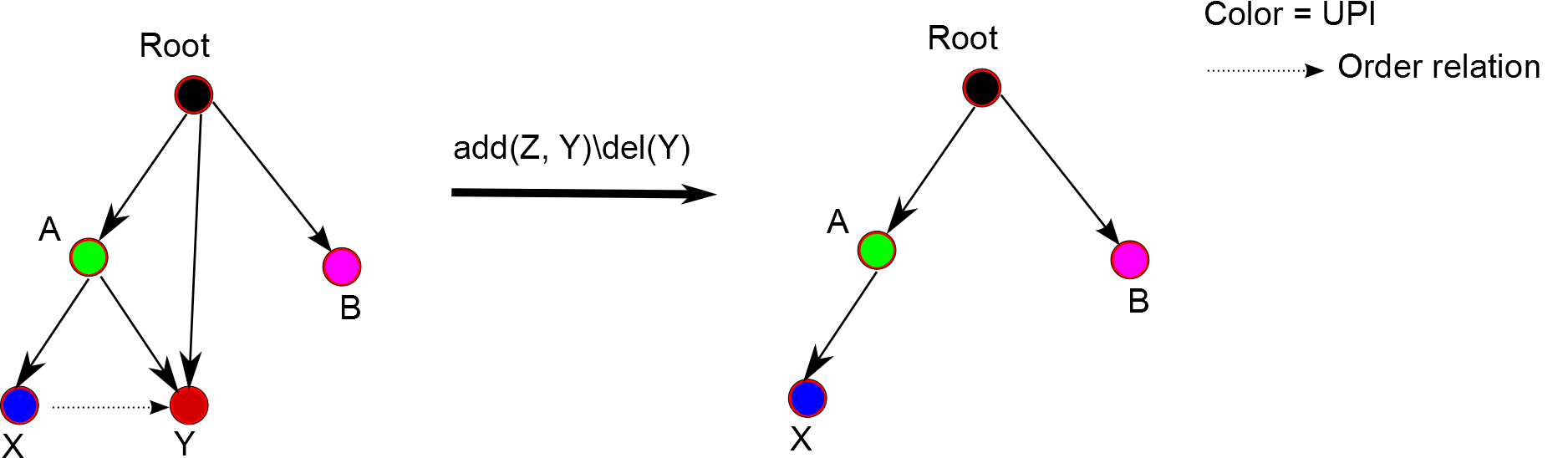} 
\caption{concurrent operations add/del with skip policy}
\label{fig:Figure_add_del} 
\end{figure} 

A tree with UPI associated to nodes can be built using any sequential
editing UPIs. However the WOOT and RGA UPI requires tombstones and
thus are more adapted to a 2P CvRDT that contains these
tombstones. For 2P CmRDT, the Logoot or Treedoc UPI approaches are
more suitable. The complexity of the children order computation
depends on the approach used. An example of such construct
is~\cite{martin09collaborative}.

\subsection{Unique positioning for Edges} 

To allow concurrent insertions on the same node at two different places
in the tree or to build edge or word tree, we propose to associate UPI
to edges. The order between the children of a node is given by the UPI
of the outgoing edge. In graph and edge trees an edge becomes a triple
$(m, n, u)$ with $m$ and $n$ two nodes and $u$ an UPI. In word trees,
a path becomes $u_1a_1u_2a_2u_2 \ldots u_{n}a_n$ with $a_i$
elements of $\Sa$ and $u_i$ UPIs. The difference between ordered trees
with edge positioning and node positioning is illustrated in
figures~\ref{fig:Diff_edge} and~\ref{fig:Diff_node}.
\begin{figure}[H] 
  \centering
\includegraphics[width=10cm]{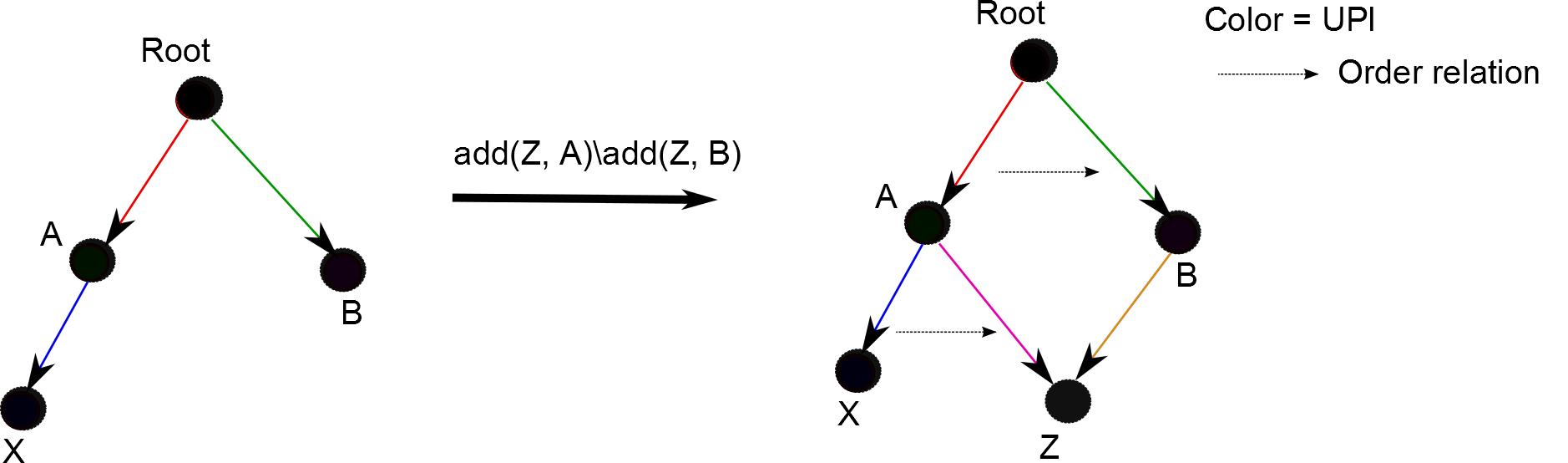} 
\caption{Concurrent operations add/add with edge positioning}
\label{fig:Diff_edge} 
\end{figure}

The formal definition of operation $rmv$ does not change and the
definition of $add$ becomes :
\begin{description}
  \item[Graph Tree]
\begin{itemize}
\item $pre(add(n,m,u), (V,E)) \equiv n \notin V \wedge m \in V \wedge
  unique(u)$
\item $post(add(n,m,u), (V,E)) \equiv n \in V \wedge (m, n, u) \in E$
\end{itemize}
  \item[Edge Tree]
\begin{itemize}
\item $pre(add(n,m,u), E) \equiv \exists (z,m,v) \in E \wedge
  unique(u)$
\item $post(add(n,m,u), E) \equiv (m, n, u) \in E$
\end{itemize}
  \item[Word Tree]
\begin{itemize}
\item $pre(add(n, p, u), T) \equiv p \in T \wedge pun \notin T \wedge
  unique(u)$
\item $post(add(n, p, u), T) \equiv pun \in T$
\end{itemize}
\end{description}

Such an edge tree is a 2E-Tree since an edge can only be added
once. And such a word tree is a 2W-Tree since a path can only be added
once.  For graph tree, we can manage node using any set CRDT to obtain
GG-Tree, 2G-Tree, LG-Tree, CG-Tree or OG-Tree. As for nodes UPI, any
sequential editing UPI can be chosen, but these are more or less
adapted to the underlying set CRDT. Logoot and Treedoc without
tombstones are more appropriate to 2x and OG CmRDT. While WOOT and RGA
are more appropriate to LG-Tree, CG-Tree and all CvRDT,

As for unordered trees, the conflicts between addition of a node and
remove of its father can be resolved using any connection
policy. Conflicts between two concurrent additions of the same element
in graph trees can be resolved using any mapping policy.

In graph tree with edge positioning, two concurrent insertion of a
node at the same place (same father and same order between children)
generates two edges (see Figure~\ref{fig:add_add_pos}). In edge tree
(and word tree), using a unique position identifier enforces to
generate two instances of the edge (and path in word tree). To allow a
different behavior, the position identifier must be non-unique.
 \begin{figure}[H]
\centering
\includegraphics[width=13cm]{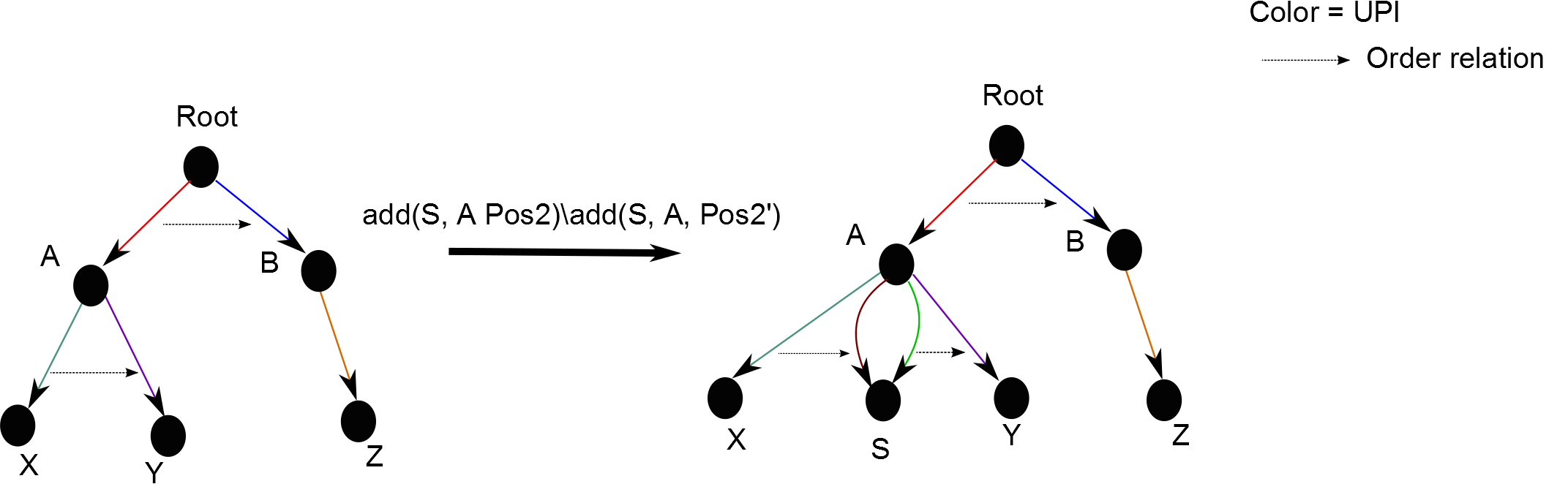} 
\caption{Two concurrent insertions at the same place with edge
  positioning}
\label{fig:add_add_pos} 
\end{figure}


\subsection{A new sequence editing CRDT : WOOTR}

Non-unique position identifiers must be totally ordered and defined
within a dense space. To obtain such properties we define a new
sequential editing CRDT called Recursive-WOOT (WOOTR).

WOOTR elements are defined inductively upon an alphabet $\Sa$ (or set
of node).
\begin{itemize}
\item $\vdash$ and $\dashv$ are elements
\item a triple $\langle a, e, f \rangle$ is an element if $a \in \Sa$
  and $e$ and $f$ are elements.
\end{itemize}

The elements $\vdash$ and $\dashv$ mark respectively the begin and the
end of a sequence. When a character $a$ is inserted between two
elements $p$ and $n$, we add the element $\langle a, p, n \rangle$. We
call $p$ the previous element and $n$ the next element of this new
element. The set of the WOOTR elements constitutes the characters
present in the sequence. The elements are ordered using the WOOT
algorithm~\cite{oster06data} assuming that elements with the same
previous and next elements are ordered using their character. For
instance, starting from an empty sequence, if a replica inserts $a$,
followed by $b$, while another replica inserts $c$ concurrently, we
obtain the set $\{ \langle a, \vdash, \dashv \rangle, \langle b,
\langle a, \vdash, \dashv \rangle, \dashv \rangle, \langle c, \vdash,
\dashv \rangle \}$ and the sequence is $abc$.

Since elements are not unique, they can be inserted concurrently by
two different replicas. However, they can also be added and removed
concurrently. Thus, as in any set, we need to manage these concurrent
operations. Eventual consistency can be achieved using a set CRDT such
as LWW-Set, CG-Set or OR-Set. Contrary to the original WOOT, we do not
require to keep deleted elements as tombstones since, when a remote
insertion occurs, the WOOT algorithm can find the place of the deleted
previous or next element before inserting the element itself. This is
particularly suitable for CmRDT OR-Set and C-Set that do not keep all
tombstones.

The size of WOOTR elements can be proportional to the size of the
document. Due to this size, such a sequential editing CRDT may not be
adapted to realtime collaborative text
editing~\cite{ahmednacer11evaluating}. However, we think that it can
be useful for trees, since in tree the element are distributed under
different fathers, the WOOTR elements grow more slowly.

\subsection{Non-unique position identifier}

With non-unique position identifiers, only one edge (or path) will be
present in the tree in case of concurrent insertion of an element at
the same place in the tree. A non-unique position identifier can be
used to order any variation of graph, edge or word tree.

For instance, the WOOTR identifier can be added to edge in graph or
edge trees. Such edges are ordered pair $(x, w)$ with the $x$ the
father node and $w$ a WOOTR element defined on the set of nodes. In
word tree, a path becomes $w_1 \ldots w_{n}$ a string of WOOTR
elements defined on the alphabet.

\section{Conclusion}
\label{sec:conclusion}

In this report, we have proposed several tree conflict-free replicated
data types (CRDT). These data types are based on set CRDTs. As any
CDRTs, tree CRDTs are eventually consistent and converge without
requiring any synchronization.

The unordered tree data types are constructed using a tree
representation (graph, edge or word), a set CRDT, one {\em connection
  policy} and one {\em mapping policy} (for graph and edge
tree). Every combination of choices is possible and is a tree
CRDT. Each of the choice correspond to the desired semantic to resolve
the two or three different conflicts between operations. The choice of
the set CRDT defines the semantic of the concurrent addition and
remove of an element. The choice of the connection policy defines the
semantic of concurrent remove of an element and addition of a
child. The choice of the mapping policy, if required, defines the
semantic of the concurrent additions of an element. With such a
construct we give to the application programmer the entire control of
the behavior of the tree CRDT.

The policies designed make some arbitrary choices to resolve the
conflicts. We think that arbitrary choices are mandatory to ensure
scalability in large-scale system. However, the application may have a
particular semantic on nodes or operations, or the final user may be
required to resolve the conflict. To facilitate such mechanism, we can
adapt the root policy and the zero policy. We can adapt the root
policy to place orphan elements under a special ``lost-and-found''
node and the zero policy to present to the application the conflicting
nodes and edge separately from the tree. 

The ordered tree data types are constructed upon unordered trees
CRDT. They consist in associating a totally ordered position
identifier to elements of the tree. These position identifier comes
from existing sequence editing CRDT and ensure eventual consistency
without synchronisation. Ordered trees share the same behavior than
the corresponding unordered tree except that a tree node can be add at
different positions under another node. The choice between the kinds
of position identifiers is a question of performance and adaptability
with the underlying set CRDT. Moreover, we introduce a new sequence
editing CRDT called WOOTR. This sequence editing CRDT is the first to
allow reintroduction of an element and to consider that concurrent
insertion of an element at the same position is the same operation.

All the combination presented can be used for any application that
require a tree. However, we think that some combination are more
adapted to some application context. For instance the unordered graph
trees are more adapted to applications managing a composite pattern or
a file system data structure. Indeed, in Unix-like files system, the
hard links allow to place a file or a repository in several different
repositories. One another hand, ordered word trees seems more adapted
to collaborative editing of structured
documents~\cite{martin10scalable}.

Finally, some constructs, especially trees builds on 2P-Set, are very
efficient, other variations and some policies, especially the several
policy in graph and edge trees, are quite costly in term of
computation complexity. We need to establish the actual scalability of
the constructs trough experimentation on realistic data set since the
actual computation cost depends highly on the degree of concurrency.

\bibliographystyle{abbrv}
\bibliography{theBib}

\end{document}